\documentclass[conference,10pt]{IEEEtran}
\usepackage{epsfig,rotating,setspace,latexsym,amsmath,epsf,amssymb,amsfonts,bm,subfigure,epstopdf, amsthm}

\usepackage{amsbsy}
\usepackage{cite,authblk}
\usepackage{bbm}
\usepackage{color, xcolor}
\usepackage{mathtools}
\usepackage{algorithm}
\usepackage{textcomp}
\usepackage[noend]{algpseudocode}
\usepackage{verbatim}
\usepackage{graphicx, graphics}
\usepackage{afterpage}
\usepackage{placeins}
\usepackage{bm}
\usepackage{authblk}
\PassOptionsToPackage{bottom}{footmisc}
\PassOptionsToPackage{symbol}{footmisc}
\usepackage{footmisc}
\usepackage{multirow}
\usepackage{booktabs}
\usepackage{float}
\usepackage{afterpage}

\DeclareMathOperator*{\argmax}{arg\,max}

\algrenewcommand\algorithmicforall{\textbf{foreach}}
\algrenewcommand\algorithmicindent{.8em}

\DeclarePairedDelimiter{\ceil}{\lceil}{\rceil}
\newtheorem{theorem}{Theorem}
\newtheorem{lemma}{Lemma}

\allowdisplaybreaks

\usepackage{graphicx}
\usepackage{epsfig}
\usepackage{epstopdf}
\begin{document}

\title{Neural Beamforming with Doppler-Aware Sparse Attention for High Mobility Environments }

\author{
Cemil Vahapoglu$^{1,2}$, Timothy J. O’Shea$^{2}$, Wan Liu$^{2}$, Sennur Ulukus$^{1}$ \\
\normalsize $^{1}$University of Maryland, College Park, MD, $^{2}$DeepSig Inc., Arlington, VA\\
\normalsize \emph{cemilnv@umd.edu, tim@deepsig.ai, wliu@deepsig.ai, ulukus@umd.edu}}

\maketitle

\begin{abstract}

    Beamforming has significance for enhancing spectral efficiency and mitigating interference in multi-antenna wireless systems, facilitating spatial multiplexing and diversity in dense and high mobility scenarios. Traditional beamforming techniques such as zero-forcing beamforming (ZFBF) and minimum mean square error (MMSE) beamforming experience significant performance deterioration under adverse channel conditions. Deep learning-based beamforming offers an alternative with nonlinear mappings from channel state information (CSI) to beamforming weights by improving robustness against dynamic channel environments. Transformer-based models are particularly effective due to their ability to model long-range dependencies across time and frequency. However, their quadratic attention complexity limits scalability in large OFDM grids. Recent studies address this issue through sparse attention mechanisms that reduce complexity while maintaining expressiveness, yet often employ patterns that disregard channel dynamics, as they are not specifically designed for wireless communication scenarios. In this work, we propose a Doppler-aware Sparse Neural Network Beamforming (Doppler-aware Sparse NNBF) model that incorporates a channel-adaptive sparse attention mechanism in a multi-user single-input multiple-output (MU-SIMO) setting. The proposed sparsity structure is configurable along 2D time-frequency axes based on channel dynamics and is theoretically proven to ensure full connectivity within $p$ hops, where $p$ is the number of attention heads. Simulation results under urban macro (UMa) channel conditions show that Doppler-aware Sparse NNBF significantly outperforms both a fixed-pattern baseline, referred to as Standard Sparse NNBF, and conventional beamforming techniques ZFBF and MMSE beamforming in high mobility scenarios, while maintaining structured sparsity with a controlled number of attended keys per query.
\end{abstract}

\section{Introduction}

Beamforming is a fundamental technique in multi-antenna wireless communication systems, employed to optimize transmission and reception patterns for improved spectral efficiency and interference mitigation. In multiple input multiple output (MIMO) systems, beamforming enables spatial multiplexing and diversity gains, thereby enhancing data rates and ensuring robust communication in dense deployment and high mobility environments. To remain effective under practical considerations, beamforming strategies are expected to be adaptive to time-varying channel conditions while suppressing inter-user interference (ISI).

Traditional beamforming strategies such as zero-forcing beamforming (ZFBF) and minimum mean square error (MMSE) beamforming rely on linear algebraic formulations utilizing available channel state information (CSI). Despite providing closed-form solutions, these methods exhibit substantial performance degradation under adverse channel conditions, such as Doppler effect or rapidly varying channel conditions, where channel estimation errors and temporal decorrelation diminish their reliability \cite{delay_performance_multiuserMISO_downlink_imperfectCSI, performance_analysis_zf_receiver_imperfectCSI}. Moreover, their computational complexity scales cubically with the number of user equipments (UEs), introducing a gap between theoretical feasibility and practical deployment \cite{erpek2020deep}.

Deep learning-based beamforming methods have gained significant attention as alternatives to traditional beamforming approaches in multi-user MIMO systems. They enable beamforming designs that directly learn complex nonlinear mappings from imperfect CSI to beamforming weights by utilizing data-driven models. Such approaches go beyond conventional solutions by capturing rich channel dynamics and adapting to diverse propagation environments, including those affected by estimation errors, hardware impairments, and temporal variabilities. Recent studies show that deep learning-based beamforming approaches can match or surpass classical techniques while offering greater flexibility in dynamic wireless scenarios \cite{Sun2018, vahapoglu2023deep, WenchaoMISODownlinkBF,DeepTx2022}.

Among deep learning architectures, transformer-based models have demonstrated remarkable success in capturing long-range dependencies across sequences. In the context of beamforming, transformers are particularly useful due to their ability to capture spatio-temporal patterns across frequency and time domains, making them well-suited for dynamic network environments such as high mobility urban macro (UMa) scenarios. Furthermore, attention mechanisms within transformers offer interpretable and adaptive feature selection, which can be utilized to mitigate ISI. Yet, a significant challenge associated with transformer-based architectures is their quadratic complexity with respect to sequence length, which limits their scalability to large OFDM grids. Recent studies on sparse attention mechanisms tackle this issue by limiting the number of attention connections for each query, facilitating efficient inference while preserving model expressiveness. Examples include strided, local sliding window, and random sparse attention patterns, which have been employed in natural language processing and computer vision \cite{child2019,longformer2020, zaheer2020bigbird}.

In this work, we propose a Doppler-aware sparse attention mechanism for neural beamforming in a multi-user single input multiple output (MU-SIMO) system setting. The proposed mechanism tailors sparsity along the 2D time-frequency axes, corresponding to embedding representation over OFDM grids. This sparsity structure is adjustable based on the temporal characteristics of the channel,  allowing the model to capture both short-range and long-range dependencies more effectively. We integrate the mechanism into our transformer-based neural network beamforming model, referred to as \textit{Doppler-aware Sparse NNBF}, and train it to maximize the average sum-rate under varying UE mobility conditions. Furthermore, we provide a theoretical proof that the proposed sparse attention structure guarantees full connectivity within $p$ hops, where $p$ is the number of attention heads. Our simulations using the UMa channel model demonstrate that Doppler-aware Sparse NNBF outperforms NNBF with standard strided attention \cite{child2019}, denoted as \textit{Standard Sparse NNBF}, as well as conventional baseline techniques under high mobility scenarios. We also empirically validate that the proposed mechanism maintains a controlled number of attended keys per query, confirming its effectiveness in enforcing structured sparsity.

\section{System Model and Problem Formulation}

\subsection{Uplink Multi-User SIMO (MU-SIMO) Setup}
We consider an uplink transmission scenario in which $N$ single antenna UEs send data streams to a base station (BS) equipped with $M$ receive antennas as shown in Fig.~\ref{System Model}.

\begin{figure}[t]
    \centerline{\includegraphics[width= 1\linewidth]{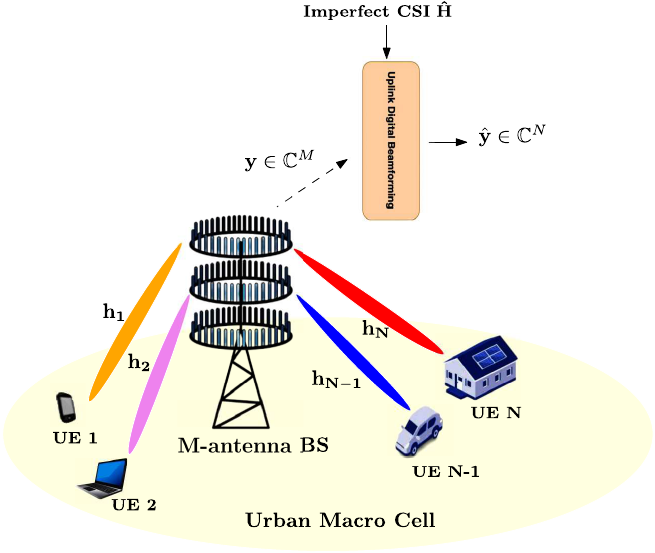}}
    \caption{Uplink multi-user SIMO system in a dense urban environment, where single-antenna UEs transmit data streams on the same time/frequency resources and the $M$-antenna BS applies digital beamforming on the received signal $\mathbf{y}$.}
    \label{System Model}
\end{figure}

The uplink channel matrix $\mathbf{H} = 
[\mathbf{h}_1  ~ \mathbf{h}_2 ~ \cdots ~ \mathbf{h}_N] \in \mathbb{C}^{M \times N}$, where $\mathbf{h}_k$ denotes the channel vector between UE $k$ and the BS. The received signal $\mathbf{y}$ can be expressed as
\begin{align} \label{received_signal}
    \mathbf{y} = \sum_{i=1}^N \mathbf{h}_i x_i +\mathbf{n}. 
\end{align}

It is presumed that ULPI-B is implemented, wherein uplink channel estimation and uplink beamforming is placed within the radio unit (RU), while the distributed unit (DU) is accountable for both uplink channel estimation and uplink equalization \cite{oranwg4}. Therefore, the uplink channel estimate $\mathbf{\hat{H}} = 
[\mathbf{\hat{h}}_1  ~ \mathbf{\hat{h}}_2 ~ \cdots ~ \mathbf{\hat{h}}_N] \in \mathbb{C}^{M \times N}$ is calculated to facilitate beamforming design directly within the RU, ensuring that all the necessary uplink processing tasks for beamforming, including the beamforming design itself, are efficiently managed locally within the RU.

The received signal in (\ref{received_signal}) is processed using beamforming weights $\mathbf{W} = [\mathbf{w}_1  ~ \mathbf{w}_2 ~ \cdots ~ \mathbf{w}_N ] \in \mathbb{C}^{M \times N}$ to retrieve data symbols while power consumption of the beamforming weights are checked to satisfy $\mathbf{w_k}^H\mathbf{w_k} \leq 1, \forall k=1, \ldots, N$. Specifically, $\mathbf{w}_k \in \mathbb{C}^{M}$ serves as the linear beamforming filter to estimate the transmitted data symbol of UE $k$, aiming to maximize throughput while mitigating the interference from other users
\begin{align}\label{received_uek}
    \mathbf{w}_k^T \mathbf{y} & = \sum_{i=1}^N \mathbf{w}_k^T\mathbf{h}_i x_i +\mathbf{w}_k^T\mathbf{n}.
\end{align}

\subsection{Beamforming Design for Sum-Rate Maximization Problem}
Our objective is to design beamforming weights that maximize the sum-rate across all UEs. The received signal for UE $k$ after applying beamformer $\mathbf{w_k}$ is
\begin{align} \label{received_uek_v2}
    \hat{y}_k &= \mathbf{w}_k^T \mathbf{y} \nonumber \\ 
     &= \underbrace{\mathbf{w}_k^T\mathbf{h}_k x_k}_{desired \, signal} + \underbrace{ \sum_{i=1, i\neq k}^N \mathbf{w}_k^T\mathbf{h}_i x_i}_{interfering \, signal} +\underbrace{\mathbf{w}_k^T\mathbf{n}}_{noise}.
\end{align}

The corresponding signal-to-interference-plus-noise ratio (SINR) for UE $k$ is
\begin{align}\label{sinr}
    \gamma_k = \frac{|\mathbf{w}_k^T \mathbf{h}_k|^2}{\sum_{i=1, i\neq k}^N |\mathbf{w}_k^T\mathbf{h}_i|^2 + \mathbb{E}|\mathbf{w}_k^T \mathbf{n}|^2}.
\end{align}

Then, the sum-rate maximization problem is 
\begin{align}\label{sum-rate maximization problem}
    \mathbf{W}^* = \argmax_{\mathbf{W}} & \quad \sum_{i=1}^N \alpha_i \log(1 + \gamma_i) \nonumber \\
    \textrm{s.t.} &\quad \textrm{tr}(\mathbf{W}^H\mathbf{W}) \leq N,
\end{align}
where $\alpha_i$ are trainable UE-specific weighting factors with $\sum_{i=1}^N \alpha_i = 1$.

\section{Deep Neural Network (DNN) Architecture}

\begin{figure*}[t]
    \centerline{\includegraphics[width=1\linewidth]{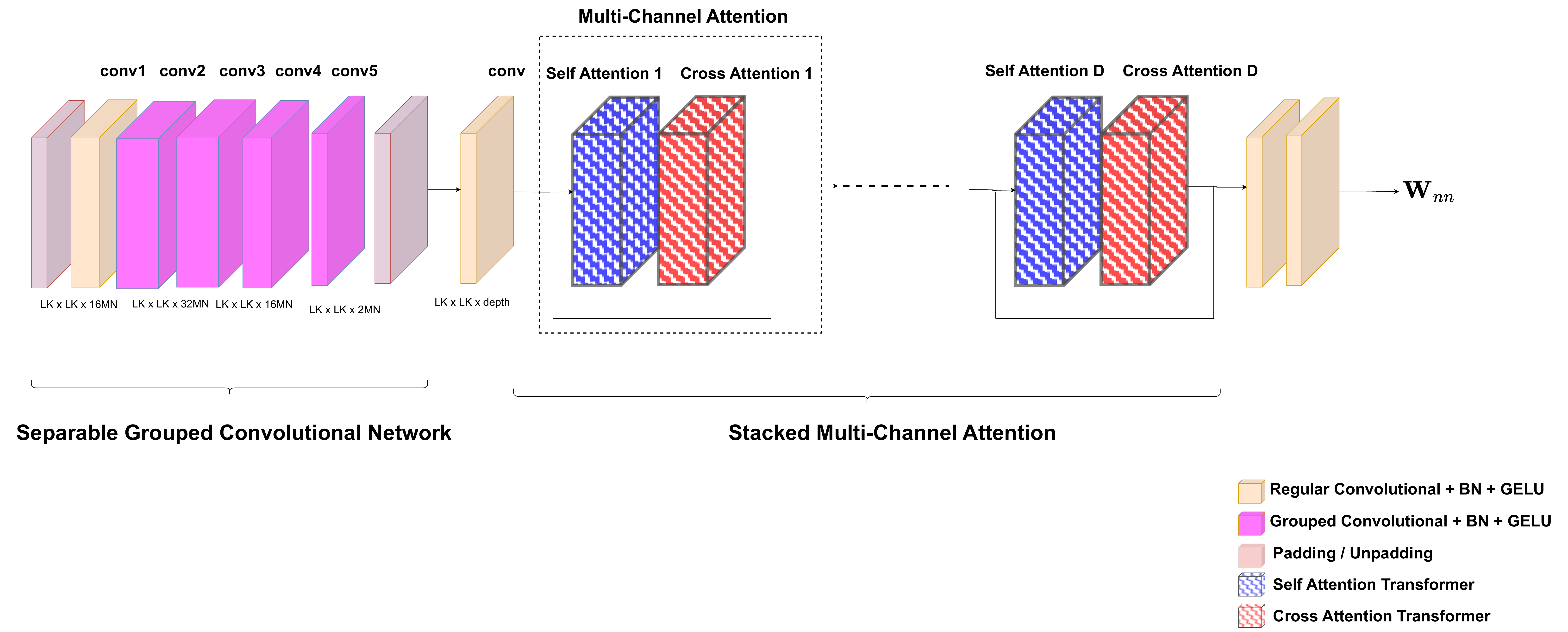}}
    \caption{Deep neural network architecture.}
    \label{DNN_architecture}
\end{figure*}

We present our DNN architecture designed to address the sum-rate maximization problem in (\ref{sum-rate maximization problem}) by learning beamforming weights $\mathbf{W}$ from the imperfect channel estimate $\hat{\mathbf{H}}$ in the frequency domain. The network takes the IQ samples of $\hat{\mathbf{H}}$ as input and outputs $\mathbf{W}$ as specified in the system model. In this context, $B$ denotes the batch size, while $L$ and $K$ represent the number of OFDM symbols and subcarriers, respectively.

\subsection{Overall Model Structure}
The DNN architecture consists of two main components: A separable grouped convolutional network and a stacked multi-channel attention module, as illustrated in Fig.~\ref{DNN_architecture}.

\subsubsection{Separable Grouped Convolutional Network}
This component processes IQ symbols of the input $\hat{\mathbf{H}}$ by reshaping it as $\mathbb{R}^{B \times 2MN \times L \times K}$. Mirror padding is first applied along $(L,K)$. The initial regular convolution is followed by multiple grouped convolutions to extract local features efficiently \cite{Chollet2016XceptionDL}. Each grouped convolution is followed by batch normalization and GELU activation. The number of groups is set as the minimum of input and output channels.

\subsubsection{Stacked Multi-Channel Attention}
This module captures correlations in both the temporal and frequency domains through self and cross attention mechanisms. Its structure follows the multi-channel attention framework in \cite{vahapoglu_transformer2025}. The only difference is the replacement of dense attention with our proposed sparse attention mechanism to reduce complexity while maintaining connectivity.

After stacked multi-channel attention module, two regular convolutional layers are used to produce the final beamforming weights $\mathbf{W}_{nn}$.

\subsection{Training Procedure}
The model is trained in an unsupervised fashion to maximize the sum-rate across all UEs. The loss is defined as
\begin{align} \label{loss_function}
    \mathcal{L}(\bm{\theta};\hat{\mathbf{H}}) = -\sum_{i=1}^N \alpha_i \log(1 + \gamma_i),
\end{align}
where $\bm{\theta}$ denotes network parameters and $\gamma_i$ is the SINR computed from the network output $\mathbf{W}_{nn} = f(\bm{\theta}; \hat{\mathbf{H}})$. Note that $\gamma_i$ depends on both the ground-truth channel $\mathbf{H}$ and estimated channel $\hat{\mathbf{H}}$, making the model robust to channel estimation errors.

For benchmarking, we compare $\mathbf{W}_{nn}$ against zero-forcing (ZF) and MMSE beamformers
\begin{align}
    \mathbf{W}_{zf} &= \left(\hat{\mathbf{H}}^H\hat{\mathbf{H}}\right)^{-1}\hat{\mathbf{H}}^H, \label{zfbf_formula} \\
    \mathbf{W}_{mmse} &= \left(\hat{\mathbf{H}}^H\hat{\mathbf{H}}+\sigma^2 \mathbf{I}_N \right)^{-1} \hat{\mathbf{H}}^H. \label{mmse_formula}
\end{align}

%%%%%%%%%%%%
\section{Doppler-Aware Sparse Attention Mechanism}

\subsection{Proposed Sparsification Structure}

We propose a structured sparsification for multi-head attention, tailored for 2D time-frequency embeddings such as OFDM resource grids. This design, referred to as \textit{Doppler-aware sparse attention} ensures full connectivity in the attention map within at most $p$ hops, where $p$ is the number of heads.

Although the proposed attention pattern operates on embedded representations rather than raw CSI values, these embeddings are produced by a separable grouped convolutional network, which is subsequently followed by positional encoding prior to the initial transformer block of Stacked Multi-Channel Attention module, as illustrated in   Fig.~\ref{DNN_architecture}. The separable grouped convolutional network maintains local time-frequency dependencies by functioning over the $(L,K)$ grid with spatially localized kernels. It ensures that embeddings are captured from local variations in OFDM symbols and subcarriers. Meanwhile, grouped convolutions along the antenna-stream dimension, as a function of $MN$, yield separate isolated spatial streams prior to their projection into a common embedding space.  This architectural design facilitates more interpretable and structured feature extraction in accordance with the wireless physical layer. Moreover, spatial indexing over $(L,K)$ is maintained through positional encoding, allowing the attention mechanism to distinguish based on their time-frequency positions. Consequently, the implementation of structured sparsity pattern is both interpretable and physically motivated, providing precise control over local and global receptive fields in the 2D attention space.

In the Doppler-aware sparse attention, global head ($h=0$) applies fixed strided attention with stride $s=\lceil T^{1-1/p} \rceil$ over the flattened 1D sequence of $T=L \cdot K$ tokens. Each token attends to all other tokens at positions offset by stride $s$ from its own modulo class
\begin{align}
    \mathbf{A}_0[i,j] =
    \begin{cases}
    1 & \text{if } j \equiv i \mod s, \\
    0 & \text{otherwise}.
    \end{cases}
\end{align}

The remaining heads ($h=1,\ldots,p-1$) employ distinct strides $(\mathrm{stride}^{(l)}_h, \mathrm{stride}^{(k)}_h)$ across time and frequency embeddings, enabling attention over time-frequency patterns that selectively capture temporal and spectral variations in embedding representations
\begin{align}
    \textrm{stride}^{(k)}_h &= \max\left(1, \left\lfloor \frac{s}{\lambda^h} \right\rfloor\right), \quad \textrm{frequency stride}, \\
    \textrm{stride}^{(l)}_h &= \max\left(1, \left\lfloor \frac{s}{\mathrm{stride}^{(k)}_h} \right\rfloor\right), \quad \textrm {time stride},      
\end{align}
where $\lambda$ is \textit{time bias} parameter, chosen based on the channel's selectivity (e.g., Doppler spread), and treated as a tunable design parameter.

For each query token $(l_q, k_q)$, with flattened query index $i=l_q \cdot K + k_q$, the attention span of head $h$ is constructed using strides $(\mathrm{stride}^{(l)}_h, \mathrm{stride}^{(k)}_h)$. Key positions $(l, k)$ are selected on the 2D grid, starting from offset positions $(\delta^{(l)}_h, \delta^{(k)}_h)$ and advanced by the corresponding strides $(\mathrm{stride}^{(l)}_h, \mathrm{stride}^{(k)}_h)$. Key positions are flattened as $j=l\cdot K+k$ to define key token indices, attended by given query token $i$ 
\begin{align}
    \delta^{(l)}_h &= (2h + i \bmod \mathrm{stride}^{(l)}_h) \bmod \mathrm{stride}^{(l)}_h, \\
    \delta^{(k)}_h &= (3h + i \bmod \mathrm{stride}^{(k)}_h) \bmod \mathrm{stride}^{(k)}_h. 
\end{align}

% plot to visualize sparsification structure
\FloatBarrier
\begin{figure}[t!]
    \vspace*{10pt}
    \centerline{\includegraphics[width= 1\linewidth]{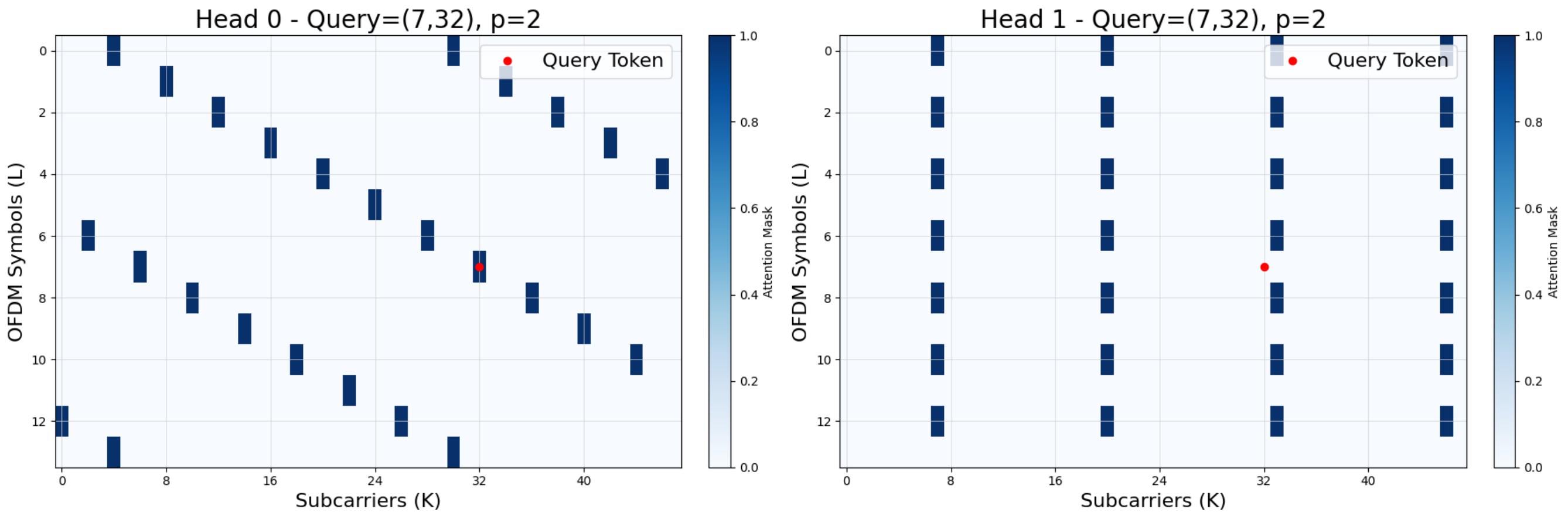}}
    \caption{Doppler-aware sparsification structure for a given query index $(l_q,k_q)=(7,32)$ when the number of heads $p$ is $2$.}
    \label{doppler-aware_sparse_pattern}
\end{figure}

\begin{figure}[t!]
    \centerline{\includegraphics[width= 1\linewidth]{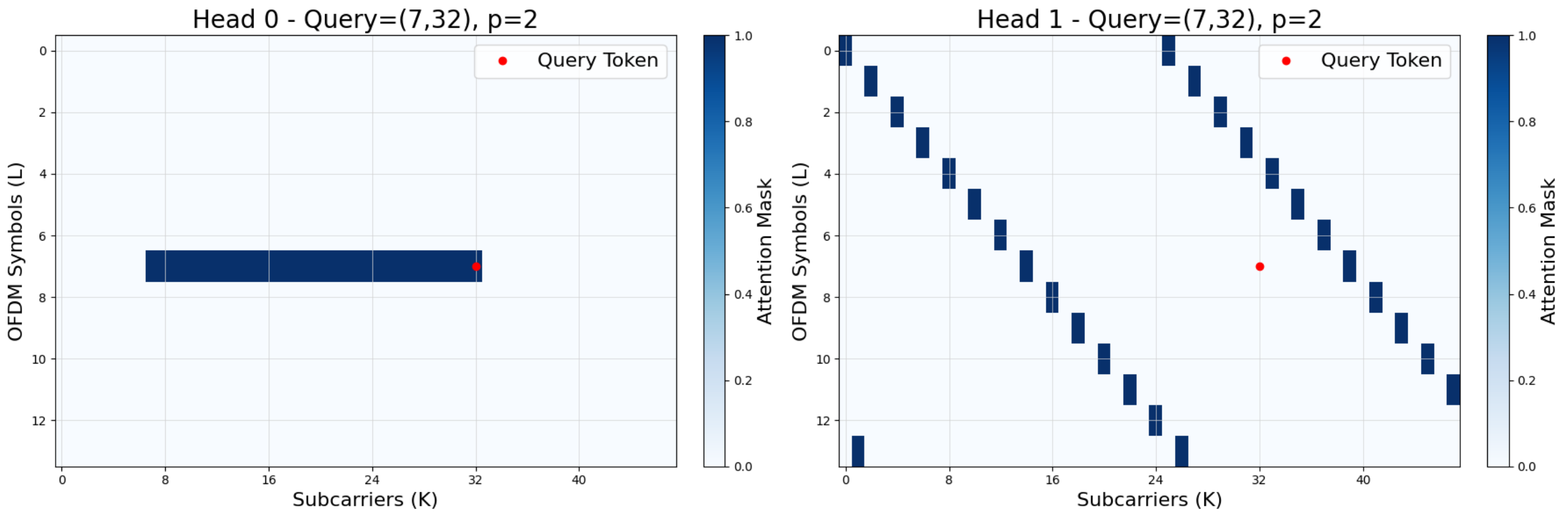}}
    \caption{Fixed strided sparsification structure \cite{child2019} for a given query index $i= 7\cdot 48 + 32$ when number of heads $p$ is 2 and fixed stride $s$ is $\ceil{T^{1-\frac{1}{p}}}$.}
    \label{fixed_strided_ref_sparsification}
\end{figure}

Proposed sparsification pattern for $p=2$ is illustrated in Fig.~\ref{doppler-aware_sparse_pattern} while fixed strided sparsification pattern \cite{child2019} is shown in Fig.~\ref{fixed_strided_ref_sparsification} for clearer comparison. The overall sparsification technique is shown in Algorithm~\ref{alg: doppler-aware sparse masks} via sparse masking design.
\begin{algorithm}[t!]
\caption{Build Doppler-Aware Sparse Masks}
\label{alg: doppler-aware sparse masks}
\textbf{Input:} Number of heads $p$, grid dimensions $(L, K)$, time bias factor $\lambda$ \\
\textbf{Output:} Sparse attention masks $\{\mathbf{A}_h\}_{h=0}^{p-1}$
\begin{algorithmic}[1]
\State $T \gets L \cdot K$, $s \gets \lceil T^{1-1/p} \rceil$
\For{$h = 0$ to $p-1$}
    \For{each query $i \in \{0, \dots, T-1\}$}
        \State $(l_q, k_q) \gets (\lfloor i / K \rfloor, i \bmod K)$
        \If{$h == 0$} \Comment{Global strided head}
            \State $r \gets i \bmod s$
            \State $\mathbf{A}_h[i, j] \gets 1$ for all $j$ s.t. $j \equiv r \pmod{s}$
        \Else
            \State $\mathrm{stride}^{(k)}_h \gets \max(1, \lfloor s / \lambda^h \rfloor)$
            \State $\mathrm{stride}^{(l)}_h \gets \max(1, \lfloor s / \mathrm{stride}^{(k)}_h \rfloor)$
            \State $\delta^{(l)}_h \gets (2h + i \bmod \mathrm{stride}^{(l)}_h) \bmod \mathrm{stride}^{(l)}_h$
            \State $\delta^{(k)} \gets (3h + i \bmod \mathrm{stride}^{(k)}_h) \bmod \mathrm{stride}^{(k)}_h$
            \For{$l = \delta^{(l)}_h$ to $L-1$ step $\mathrm{stride}^{(l)}_h$}
                \For{$k = \delta^{(k)}_h$ to $K-1$ step $\mathrm{stride}^{(k)}_h$}
                    \State $j \gets l \cdot K + k$
                    \State $\mathbf{A}_h[i, j] \gets 1$
                \EndFor
            \EndFor
        \EndIf
    \EndFor
\EndFor
\end{algorithmic}
\end{algorithm}

\subsection{Multi-Head Attention Graph}

The attention pattern of each head $h \in \{0, 1, \dots, p{-}1\}$, characterized by binary attention masks $\mathbf{A}_h \in \{0,1\}^{T \times T}$ where $\mathbf{A}_h[i,j] = 1$ indicates that query token $i$ attends to key token $j$ through head $h$, forms a directed attention graph $\mathcal{G}_h = (\mathcal{V}, \mathcal{E}_h)$, $\forall h \in \{0, 1, \dots, p{-}1\}$, where $V = \{0, 1, \dots, T{-}1\}$ represents the tokens and $\mathcal{E}_h = \{(i,j) \mid \mathbf{A}_h[i,j]=1 \}$ denotes the directed edges for head $h$. 
Consequently, each head is associated with an individual layer of edges when the multi-head attention graph corresponds to the union of edges
\begin{align}
    \mathcal{G} &= \bigcup_{h=0}^{p-1} \mathcal{G}_h = \bigcup_{h=0}^{p-1} \left(\mathcal{V},\mathcal{E}_h \right).
\end{align}

\begin{lemma}[Partitioning by Global Head]
\label{lemma:modulo_partition}
Let $s = \lceil T^{1-1/p} \rceil$ denote the stride of the global head ($h = 0$). Then, the attention graph $\mathcal{G}_0$ corresponding to Head 0 partitions the node set $\mathcal{V} = \{0, 1, \dots, T{-}1\}$ into $s$ disjoint equivalence classes:
\[
C_r = \left\{ i \in \mathcal{V} \,\middle|\, i \equiv r \!\!\! \mod s \right\}, \quad r = 0, 1, \dots, s{-}1.
\]
Each equivalence class $C_r$ generates a complete subgraph in $\mathcal{G}_0$. There are no edges between nodes of distinct classes, i.e., $\mathcal{G}_0$ contains no inter-class connections.
\end{lemma}

\begin{proof}
    By construction of Algorithm~\ref{alg: doppler-aware sparse masks}, token $i$ attends to tokens $j$ that meets the condition $j \equiv i \bmod s$ under Head 0. As a result, for each class $C_r$, any $i, j \in C_r$ satisfies $i \equiv j \bmod s$ and are mutually accessible, establishing a complete subgraph. Conversely, if $i \in C_r$ and $j \notin C_r$, then $j \not\equiv i \bmod s$, resulting in $\mathbf{A}_0[i,j] = 0$. Therefore, there are no inter-class edges in $\mathcal{G}_0$.
\end{proof}

\begin{theorem}[Full Connectivity with Global Head]
\label{thm:full_connectivity}
Let $\mathcal{G}_0, \dots, \mathcal{G}_{p-1}$ be the attention graphs generated by $p$ attention heads over a sequence of $T = L \cdot K$ tokens organized in a $L \times K$ 2D grid. Suppose that the global head, denoted as Head $h = 0$, apply fixed strided attention with stride $s = \lceil T^{1-1/p} \rceil$, resulting in $s$ disjoint equivalence classes $C_r = \left\{ i \in \mathcal{V} \,\middle|\, i \equiv r \!\!\! \mod s \right\}, \, r = 0, 1, \dots, s-1$.

Then, the union graph $\mathcal{G} = \bigcup_{h=0}^{p-1} \mathcal{G}_h$ is fully connected, i.e., there exists a path between every pair of tokens in at most $p$ hops, provided that there exists at least one head $h \in \{1, \dots, p-1\}$ using strides $(\mathrm{stride}^{(l)}_h, \mathrm{stride}^{(k)}_h) \in \mathbb{N}^2$ such that
    \[
        \gcd\left( \gcd(\mathrm{stride}^{(l)}_h \cdot K,\ \mathrm{stride}^{(k)}_h),\ s \right) = 1
    \]
\end{theorem}

\begin{proof}[Proof Sketch]

We consider graph structure in two steps: intra-class and inter-class connectivity. Then, theory of linear congruences is utilized to prove connectivity guarantee.

\paragraph{Intra-class connectivity} By Lemma~\ref{lemma:modulo_partition}, the global head $\mathcal{G}_0$ partitions the node set $\mathcal{V}$ into $s$ disjoint equivalence classes $C_0, C_1, \ldots, C_{s-1}$, where each class forms a complete subgraph and no edges exist between different classes.

\paragraph{Inter-class bridging} Let head $h$, $h\geq1$, use strides $(\mathrm{stride}^{(l)}_h, \mathrm{stride}^{(k)}_h)$. For a query token at $(l_q, k_q)$, the attended keys $(l, k)$ satisfy
\[
l = l_q + n \cdot \mathrm{stride}^{(l)}_h,\quad
k = k_q + m \cdot \mathrm{stride}^{(k)}_h,  \quad n, m \in \mathbb{Z}_{\geq 0}.
\]
This maps to flattened key indices as follows.
\[
j = l \cdot K + k = \underbrace{l_q K + k_q}_{i} + n \cdot \mathrm{stride}^{(l)}_h \cdot K + m \cdot \mathrm{stride}^{(k)}_h.
\]
Consequently, the set of attention offsets with respect to query index $i$ is
\[
\mathcal{S}_h = \left\{ n \cdot \mathrm{stride}^{(l)}_h \cdot K + m \cdot \mathrm{stride}^{(k)}_h \mid n, m \in \mathbb{Z}_{\geq 0} \right\}.
\]
By defining the effective flattened step as,
\[
P_h = \gcd\left( \mathrm{stride}^{(l)}_h \cdot K,\ \mathrm{stride}^{(k)}_h \right).
\]
Then, $\mathcal{S}_h = \{ t \cdot P_h \mid t \in \mathbb{Z}_{\geq 0} \}$.

\paragraph{Inter-Class Connectivity via Linear Congruence}
Assume $\gcd(P_h,\ s) = 1$. For any $i \in \mathcal{C}_r$, consider the set of indices reachable via $t \cdot P_h$ steps
\[
i + t \cdot P_h \mod T.
\]
Then, for each $r' \in \{0, \dots, s-1\}$, there exists $t$ such that
\[
i + t \cdot P_h \equiv r' \mod s.
\]
This follows directly from the existence of solutions of linear congruences (see Theorem~4.7 in \cite{burton2007number}).

Therefore, head $h$ bridges all equivalence classes, ensuring that tokens from different $\mathcal{C}_r$ can be reached.
\end{proof}

\section{Experiments}
In our experiments, we evaluate the performance of NNBF under varying UE mobility conditions to simulate different Doppler effects. Our proposed approach, referred to as \textit{Doppler-aware sparse NNBF}, is compared against NNBF using a fixed strided attention mechanism \cite{child2019}, denoted as \textit{standard sparse NNBF}, as well as baseline methods ZFBF and MMSE beamforming. Spectral efficiency (in bps/Hz) and BLER are used as performance metrics to assess throughput and communication reliability.

\subsection{System and Training Specifications}

\FloatBarrier
\begin{table} [t]
    \vspace*{10pt}
    \begin{center}
    \resizebox{\columnwidth}{!}{%
    \begin{tabular}{| c | c |}
    \hline
    \textbf{Parameter} & \textbf{Value} \\ \hline
    Number of resource blocks (RBs) & 4 (48 subcarriers)  \\ \hline
    Maximum Doppler shift $f_d$ & 1040 Hz \\ \hline
    UE velocity $[v_{\min}, v_{\max}] \mathrm{(m/s)}$ & [0,10],[30,40] \\ \hline
    Carrier frequency $f_c$ & 2.6 GHz \\ \hline
    Subcarrier spacing & 30 kHz \\ \hline
    Transmission time interval (TTI) & 500 $\mu s$ \\ \hline
    Coding rate & $\frac{3}{4}$ \\ \hline
    Modulation scheme & 16QAM \\ \hline
    Training SNR & [-10,20] dB   \\\hline
    Learning rate & $[10^{-5}, 10^{-2}]$ \\ \hline
    $\alpha_{\textrm{la}}$ & 0.5 \\ \hline
    k & 13 \\ \hline
    Minimum training SNR ranges & [15,20],[10,15],[5,10],[0,5],[-10,0] \\ \hline
    \end{tabular}}
    \end{center}
    \caption{System \& training parameters.}
    \label{table:system parameters}
\end{table}

Experiments are conducted with $2\times8$ antenna configurations, denoted as $N \times M$. Models are trained over a broad SNR range of $[-10, 20]$ dB to cover both low and high SNR operating regimes. Channel responses are generated using the UMa channel model in the NVIDIA Sionna library \cite{sionna}, following the 3GPP TR 38.901 specifications \cite{3gppTR38901}.

For training, hyperparameter optimization is performed using Optuna \cite{optuna_2019} across various optimizers $\{$Adam, AdamW, RAdam, RMSprop, Adagrad, Adadelta$\}$ and learning rate schedulers $\{$ReduceLROnPlateau, CosineAnnealing, CosineAnnealingWarmRestarts, ExponentialLR, CyclicLR$\}$. The Lookahead optimizer is also employed to enhance convergence and training stability, with $k = 13$ update steps and an interpolation coefficient of $\alpha_{\textrm{la}} = 0.5$. Specifically, the fast weights $\bm{\theta}^{\textrm{fast}}_t$ are updated for $k$ steps using the base optimizer, after which the slow weights are updated as
\begin{align}
    \bm{\theta}^{\textrm{slow}}_t = \bm{\theta}^{\textrm{slow}}_{t-k} + \alpha_{\textrm{la}}(\bm{\theta}^{\textrm{fast}}_t - \bm{\theta}^{\textrm{slow}}_{t-k}).
\end{align}

A curriculum learning strategy is adopted, where training progresses from easier to more challenging tasks by gradually lowering the minimum SNR. The maximum SNR is fixed at 20 dB, while the minimum SNR at each stage is treated as a tunable hyperparameter. All system and training parameters are summarized in Table~\ref{table:system parameters}.

\subsection{Results and Analysis}

Figs.~\ref{fig:low-doppler} and~\ref{fig:high-doppler} present the performance of the proposed Doppler-aware sparse NNBF under low and high UE mobility conditions, respectively. In the low-Doppler scenario $[v_{\min}, v_{\max}] = [0, 10]\,\mathrm{m/s}$, both Doppler-aware and standard sparse NNBF exhibit comparable performance, achieving similar sum-rate and BLER performances as ZFBF and MMSE beamforming. However, under high Doppler conditions $[v_{\min}, v_{\max}] = [30, 40]\,\mathrm{m/s}$, baseline techniques experience significant degradation, while Doppler-aware sparse NNBF outperforms all other methods in both spectral efficiency and BLER. 

These results highlight the importance of designing attention mechanisms that are adaptive to channel dynamics. The performance gap observed in high-mobility scenarios demonstrate that fixed strided attention inadequately captures rapidly changing frequency-time dependencies, whereas Doppler-aware sparsification facilitates more robust feature extraction.

\begin{figure}[t]
    \begin{center}
     	\subfigure[]{%
     	\includegraphics[width=0.49\linewidth]{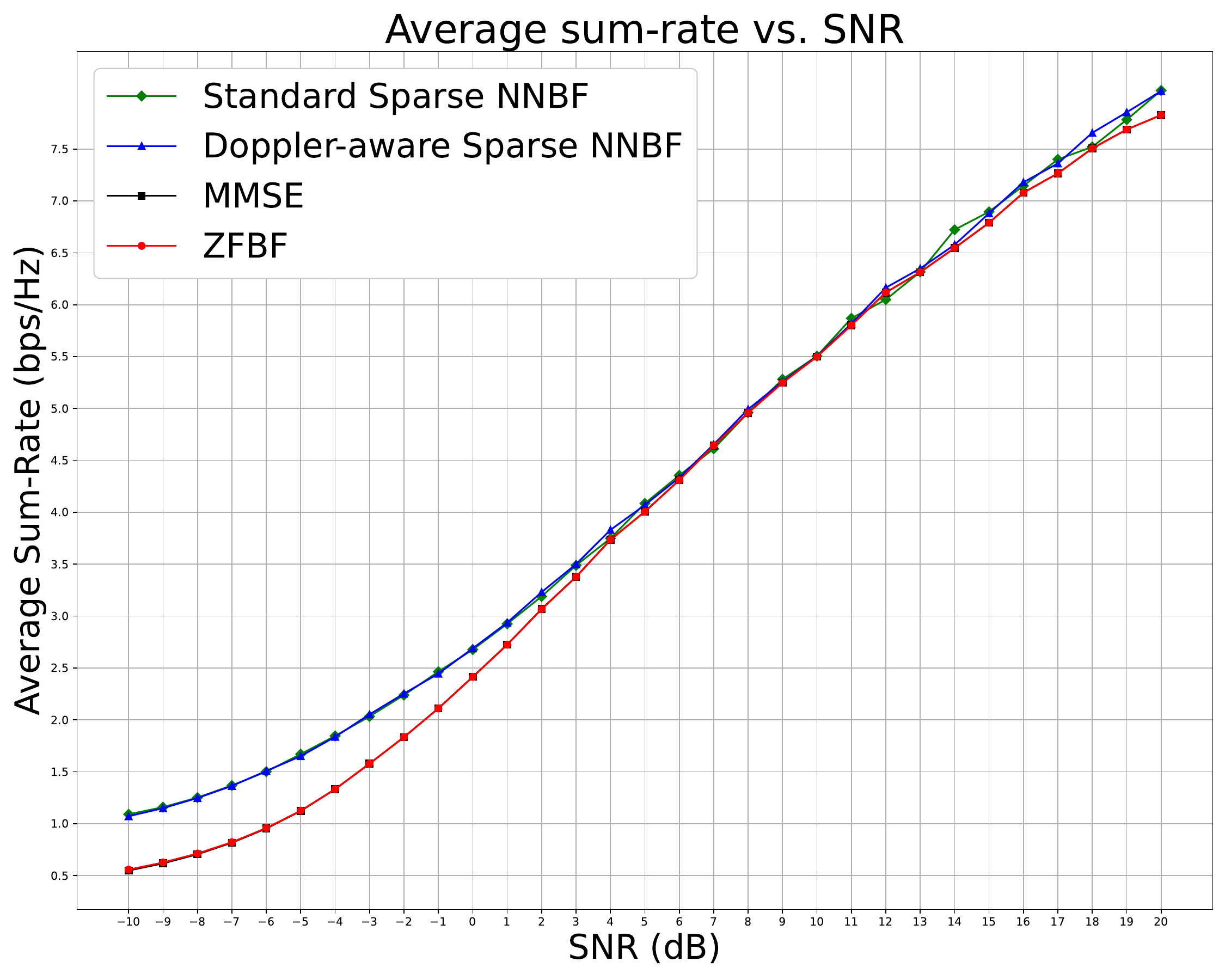}}
     	\subfigure[]{%
     	\includegraphics[width=0.49\linewidth]{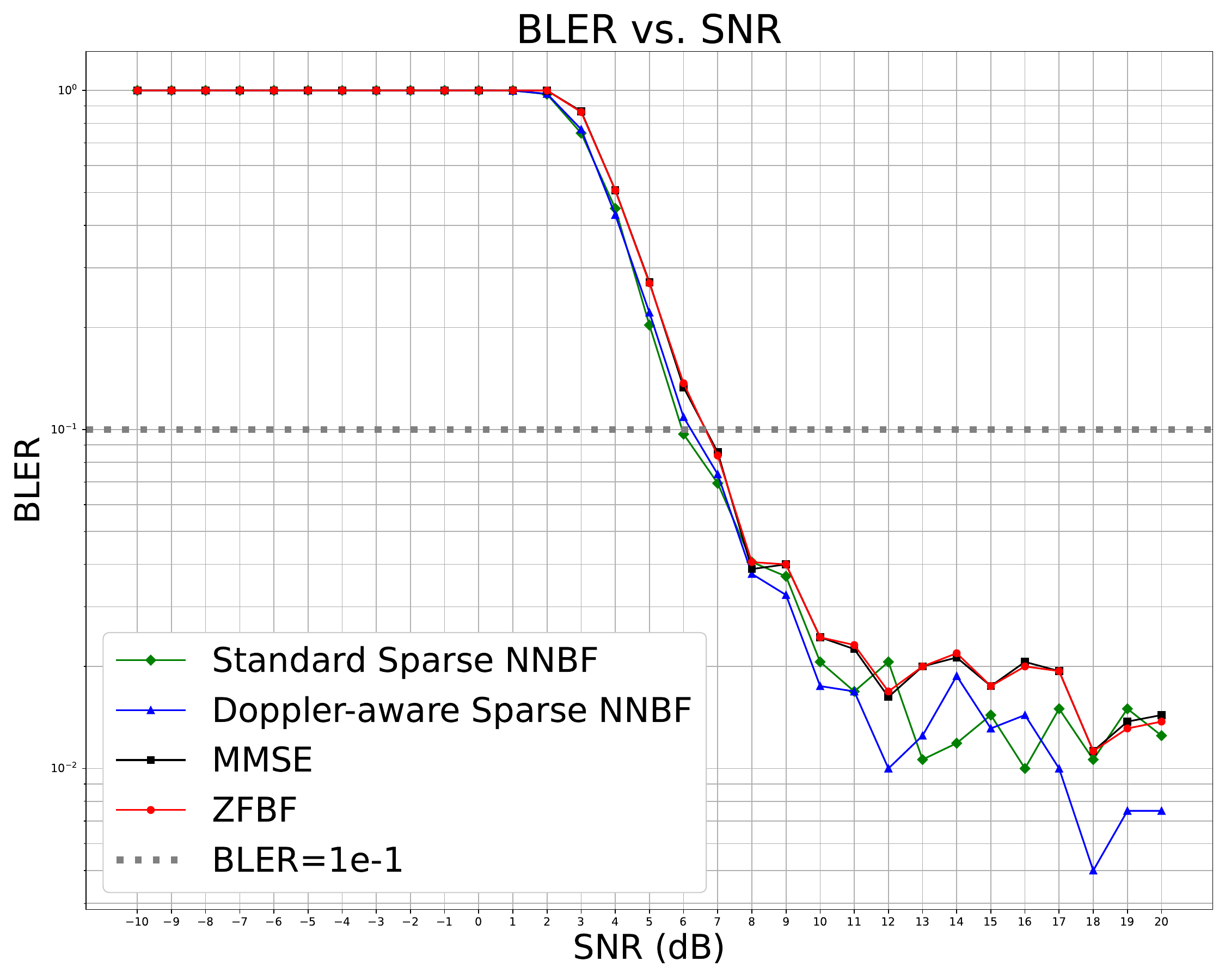}}
    \end{center}
     \caption{Performance comparison under low Doppler conditions \([v_{\min}, v_{\max}] = [0,10]\,\mathrm{m/s}\). Doppler-aware and standard sparse NNBF methods perform similarly and match baseline methods ZFBF and MMSE in both (a) average sum-rate and (b) BLER.} 
    \label{fig:low-doppler}
\end{figure}

\begin{figure}[h!]
    \begin{center}
     	\subfigure[]{%
     	\includegraphics[width=0.49\linewidth]{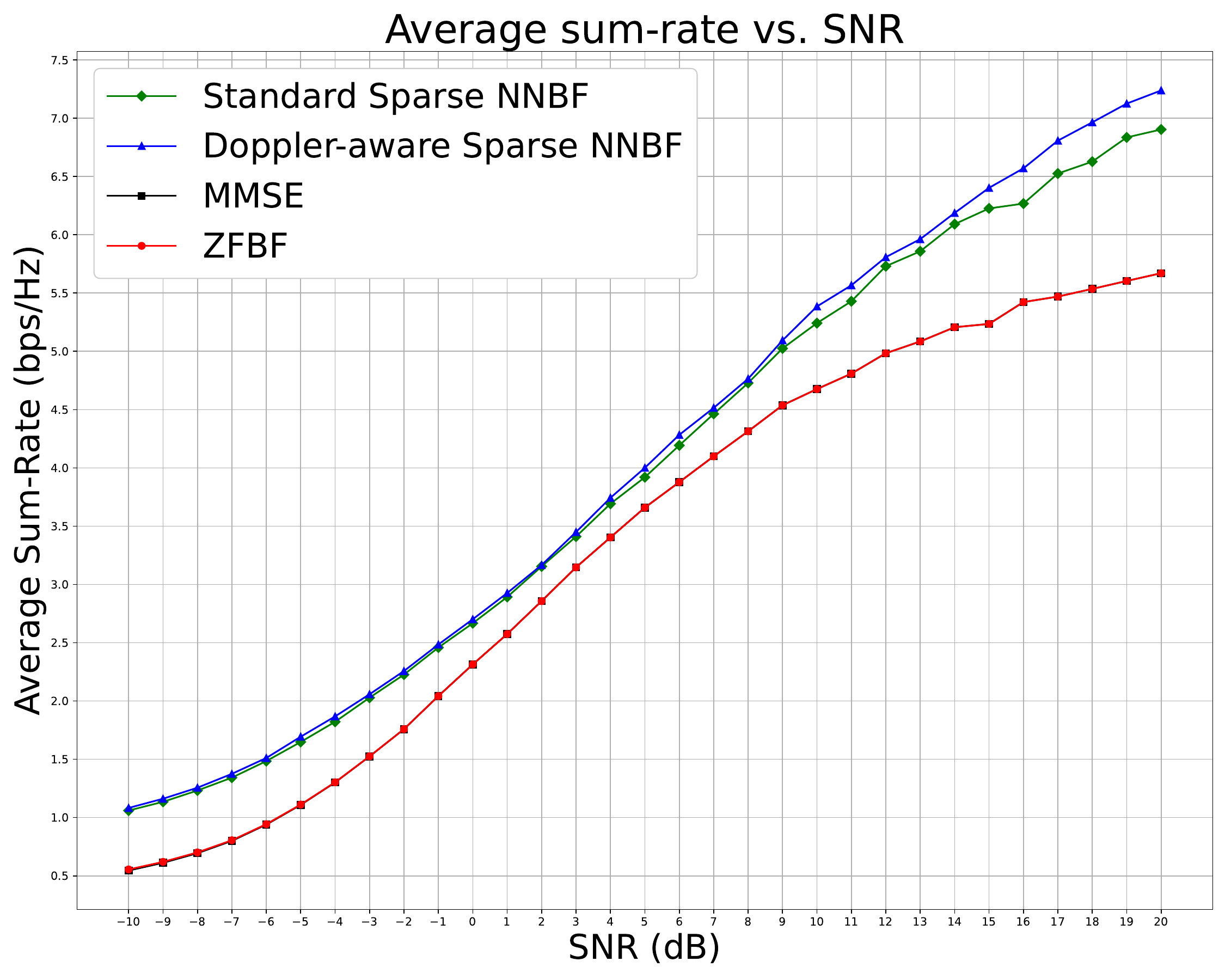}}
     	\subfigure[]{%
     	\includegraphics[width=0.49\linewidth]{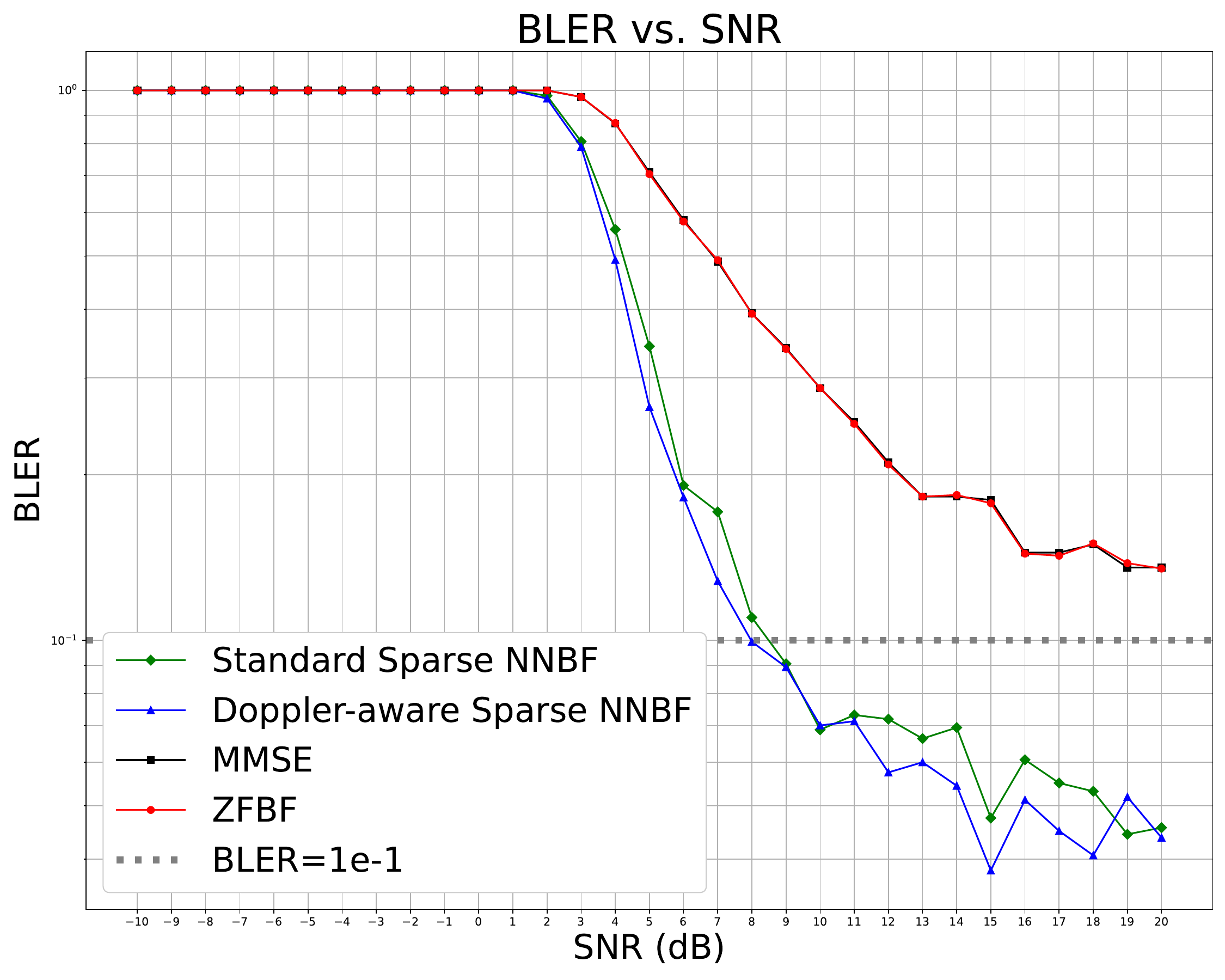}}
    \end{center}
     \caption{Performance comparison under high Doppler conditions \([v_{\min}, v_{\max}] = [30,40]\,\mathrm{m/s}\). Doppler-aware sparse NNBF outperforms standard sparse NNBF and traditional baselines ZFBF and MMSE in both (a) average sum-rate and (b) BLER.} 
    \label{fig:high-doppler}
\end{figure}

\begin{figure}[h!]
    \centerline{\includegraphics[width= 1\linewidth]{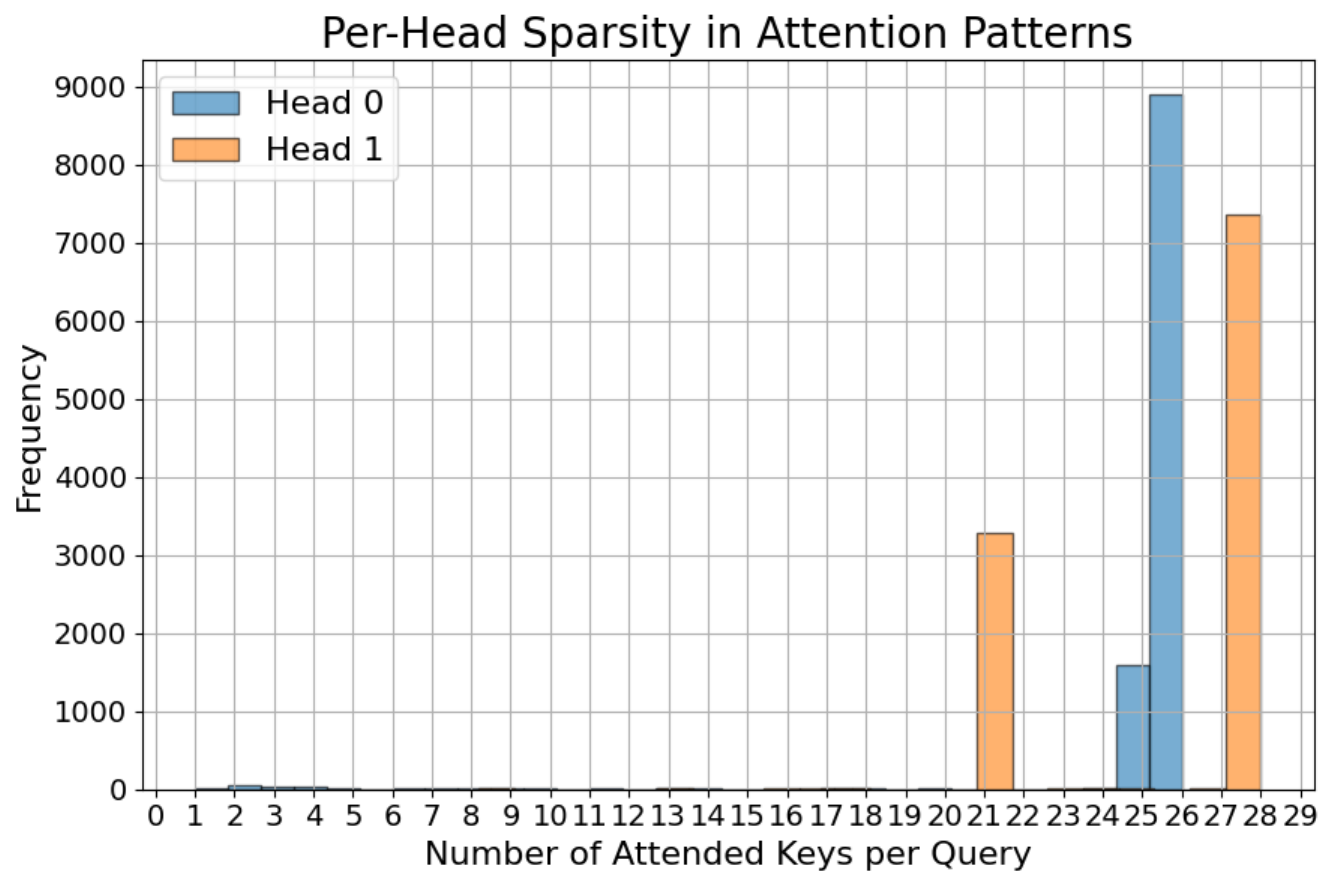}}
    \caption{Histogram of the number of attended keys per query for each head. A total of 10,752 queries are evaluated, corresponding to a batch size of 16 and a sequence length of 672 queries per sample. The sharp peak confirms that each query attends to a fixed or narrow range of keys, consistent with the proposed sparsity pattern.}
    \label{fig:histogram}
\end{figure}

Fig.~\ref{fig:histogram} shows the histogram of the number of attended keys per query for each attention head, based on attention scores saved during training. The distribution verifies that the proposed sparsification strategy maintains a controlled number of active attention connections, ensuring both computational efficiency and full query coverage as intended.

\bibliographystyle{unsrt}
\bibliography{main}

\begin{thebibliography}{10}

\bibitem{delay_performance_multiuserMISO_downlink_imperfectCSI}
S.~Schiessl, J.~Gross, M.~Skoglund, and G.~Caire.
\newblock Delay performance of the multiuser {MISO} downlink under imperfect {CSI} and finite-length coding.
\newblock {\em IEEE Journal on Selected Areas in Communications}, 37(4):765--779, April 2019.

\bibitem{performance_analysis_zf_receiver_imperfectCSI}
V.~D. Nguyen and O.~S. Shin.
\newblock Performance analysis of zf receivers with imperfect {CSI} for uplink massive {MIMO} systems.
\newblock 2016.
\newblock Available online at arXiv:1606.03150.

\bibitem{erpek2020deep}
T.~Erpek, T.~J. O'Shea, Y.~E. Sagduyu, Y.~Shi, and T.~C. Clancy.
\newblock Deep learning for wireless communications.
\newblock 2020.
\newblock Available online at arXiv:2005.06068.

\bibitem{Sun2018}
H.~Sun, X.~Chen, Q.~Shi, M.~Hong, X.~Fu, and N.~D. Sidiropoulos.
\newblock Learning to optimize: Training deep neural networks for interference management.
\newblock {\em IEEE Transactions on Signal Processing}, 66(20):5438--5453, October 2018.

\bibitem{vahapoglu2023deep}
C.~Vahapoglu, T.~J. O'Shea, T.~Roy, and S.~Ulukus.
\newblock Deep learning based uplink multi-user {SIMO} beamforming design.
\newblock In {\em IEEE ICMLCN}, May 2023.

\bibitem{WenchaoMISODownlinkBF}
W.~Xia, G.~Zheng, Y.~Zhu, J.~Zhang, J.~Wang, and A.P. Petropulu.
\newblock A deep learning framework for optimization of {MISO} downlink beamforming.
\newblock {\em IEEE Transactions on Communications}, 68(3):1866--1880, March 2020.

\bibitem{DeepTx2022}
J.~Huttunen, D.~Korpi, and M.~Honkala.
\newblock {DeepTx}: Deep learning beamforming with channel prediction.
\newblock {\em IEEE Transactions on Wireless Communications}, 22(3):1855--1867, March 2023.

\bibitem{child2019}
R.~Child, S.~Gray, A.~Radford, and I.~Sutskever.
\newblock Generating long sequences with sparse transformers.
\newblock {\em arXiv preprint arXiv:1904.10509}, 2019.

\bibitem{longformer2020}
I.~Beltagy, M.~E. Peters, and A.~Cohan.
\newblock Longformer: The long‑document transformer.
\newblock 2020.
\newblock Available online at arXiv:2004.05150.

\bibitem{zaheer2020bigbird}
M.~Zaheer, G.~Guruganesh, A.~Dubey, J.~Ainslie, C.~Alberti, S.~Ontañón, P.~Pham, A.~Ravula, Q.~Wang, L.~Yang, and A.~Ahmed.
\newblock Big bird: Transformers for longer sequences.
\newblock In {\em NeurIPS}, 2020.

\bibitem{oranwg4}
{{O-RAN} Fronthaul Control, User and Synchronization Plane Specification}.
\newblock Technical Report ORAN.WG4.CUS.0-v07.02, O-RAN Alliance, 2023.
\newblock O-RAN WG4 CUS Specification v7.02.

\bibitem{Chollet2016XceptionDL}
F.~Chollet.
\newblock Xception: Deep learning with depthwise separable convolutions.
\newblock In {\em CVPR}, July 2017.

\bibitem{vahapoglu_transformer2025}
C.~Vahapoglu, T.~J. O'Shea, W.~Liu, T.~Roy, and S.~Ulukus.
\newblock Transformer-driven neural beamforming with imperfect csi in urban macro wireless channels.
\newblock 2025.
\newblock Available online at arXiv:2504.11667.

\bibitem{burton2007number}
D.~M. Burton.
\newblock {\em Elementary Number Theory}.
\newblock McGraw-Hill, 6th edition, 2007.

\bibitem{sionna}
J.~Hoydis, S.~Cammerer, F.~{Ait Aoudia}, A.~Vem, N.~Binder, G.~Marcus, and A.~Keller.
\newblock Sionna: An open-source library for next-generation physical layer research.
\newblock {\em arXiv preprint}, March 2022.

\bibitem{3gppTR38901}
{3GPP}.
\newblock {Study on channel model for frequencies from 0.5 to 100 GHz}.
\newblock Technical Report TR 38.901, {3rd Generation Partnership Project (3GPP)}, April 2022.
\newblock Version 17.0.0.

\bibitem{optuna_2019}
T.~Akiba, S.~Sano, T.~Yanase, T.~Ohta, and M.~Koyama.
\newblock Optuna: A next-generation hyperparameter optimization framework.
\newblock In {\em KDD}, 2019.

\end{thebibliography}
\end{document}